\documentclass[12pt]{article}

\usepackage{amsthm}
\usepackage{graphicx}
\usepackage{bm}
\usepackage{amsmath}
\usepackage{amssymb}

\newtheorem{theorem}{Theorem}
\newtheorem{lemma}[theorem]{Lemma}
\newtheorem{corollary}[theorem]{Corollary}
\newtheorem{claim}[theorem]{Claim}

\begin{document}

\title{%
A Relation between the Protocol Partition 
Number and the Quasi-Additive Bound}
\author{Naoyuki Kamiyama\thanks{%
E-mail: {\ttfamily kamiyama@ise.chuo-u.ac.jp}.
}\\
Department of Information and System Engineering,\\
Chuo University, Japan.
}
\date{}

\maketitle

\begin{abstract}
\noindent
In this note, we show that the linear programming for 
computing the quasi-additive bound 
of the formula size of a Boolean function 
presented by Ueno [MFCS'10] is equivalent to 
the dual problem of the linear programming relaxation 
of an integer programming for computing
the protocol partition number. 
Together with the result of Ueno [MFCS'10], 
our results imply that 
there exists no gap between 
our integer programming for computing
the protocol partition number and 
its linear programming relaxation. 
\end{abstract}

\section{Introduction}

Proving lower bounds for a concrete computational model
is a fundamental problem in the computational complexity theory. 
In this note, we consider formula size lower bounds for 
a Boolean function. Karchmer and Wigderson~\cite{KW90} shown 
that the size of a smallest formula computing a Boolean function $f$ 
is equal to the {\it protocol partition number}
of the {\it communication matrix} arising from $f$. 
Karchmer, Kushilevitz and Nisan~\cite{KKN95}
formulated the problem of computing 
a lower bound for a protocol partition number 
as an integer programming problem and 
introduced a technique, called the {\it rectangle bound},
which gives a lower bound by showing a feasible solution of 
the dual problem of its linear programming relaxation. 
However, Karchmer, Kushilevitz and Nisan~\cite{KKN95}
also showed that this technique can not prove a lower bound 
larger than $4n^2$ for non-monotone formula size in general. 

Recently, Ueno~\cite{U10} introduced a novel technique, called 
the {\it quasi-additive bound}, which is inspired by the notion 
of {\it subadditive rectangle measures} presented by 
Hrube\v{s}, Jukna, Kulikov and Pudl\'{a}k~\cite{HJKP10}. 
Although the linear programming for computing the quasi-additive
bound can be seen as a simple extension of the linear programming
for computing the rectangle bound, Ueno~\cite{U10} showed that 
the quasi-additive bound can surpass the rectangle bound and 
it is potentially strong enough to give the matching formula
size lower bounds. 

In this note, we show that the linear programming for computing 
the quasi-additive bound of the formula size of a Boolean function 
presented by Ueno~\cite{U10} is equivalent to 
the dual problem of the linear programming relaxation 
of an integer programming for computing
the protocol partition number. 
Together with the result of Ueno~\cite{U10}, 
our results imply that 
there exists no gap between 
our integer programming for computing
the protocol partition number and 
its linear programming relaxation.
We hope that the results of this note help 
to understand why the quasi-additive bound is more 
powerful than the rectangle bound. 
Furthermore, to the best of our knowledge, no one 
studied an exact integer programming formulation for 
computing a protocol partition number. 
Thus, it may be of independent interests. 

\section{Preliminaries}

Let $\mathbb{R}$ and $\mathbb{N}$ be the sets of 
reals and non-negative integers, respectively. 
Given a vector $x$ on a ground set $U$, we use the 
notation $|x|=\sum_{u \in U}x_u$. 
A {\it relation} $T$ is a non-empty subset of $X\times Y \times Z$ 
for some finite sets $X$, $Y$ and $Z$. 
When we emphasize that a relation $T$ is a subset of 
$X\times Y \times Z$, 
we say that $T$ is a relation on $(X,Y,Z)$.  
In this note, we assume that for each relation $T$ on 
$(X,Y,Z)$ and $(x,y) \in X\times Y$ 
there exists $z \in Z$ such that $(x,y,z)\in T$. 

A {\it formula} is a binary tree with each leaf labeled
by a literal and each non-leaf vertex labeled by either 
of the binary connectives 
$\vee$ and $\wedge$. A literal is either a variable or 
its negation. The size of a formula is its
number of literals. 
For a Boolean function $f$, we define formula size $L(f)$ as the
size of a smallest formula computing $f$.

Karchmer and Wigderson~\cite{KW90} characterized the size of a smallest 
formula computing a Boolean function by using the notions 
of a {\it communication matrix} and a {\it protocol partition number}.
Suppose that we are given a relation $T$ on $(X,Y,Z)$. 
The communication matrix $M_T$ of $T$
is defined by a matrix whose rows and columns are indexed by 
$X$ and $Y$ respectively. 
Furthermore, each cell $(x,y) \in X \times Y$ 
of $M_T$ contains $z \in Z$ such that 
$(x,y,z) \in T$. 

A {\it  rectangle} 
of $M_T$ is a nonempty direct product 
$X' \times Y' \subseteq X\times Y$.
A  rectangle $X'\times Y'$
is called {\it monochromatic} if there exists  
$z \in Z$ such that $(x,y,z) \in T$ 
for all $(x,y) \in X'\times Y'$.
For a  rectangle 
$X' \times Y'$, a {\it partition} of $X'\times Y'$ is 
a pair of  rectangles 
$X'_1 \times Y'$ and $X'_2 \times Y'$ such that 
$X'=X'_1\cup X'_2$ and $X'_1 \cap X'_2=\emptyset$, 
or a pair of  rectangles 
$X' \times Y'_1$ and 
$X' \times Y'_2$ such that 
$Y'=Y'_1\cup Y'_2$ and $Y'_1 \cap Y'_2=\emptyset$. 

Suppose that we are give a set $\mathcal{R}$ of disjoint 
 rectangles. 
We say that $\mathcal{R}$ {\it recursively} partitions $M_T$ if 
$\cup_{R \in \mathcal{R}}R=M_T$ and there exists a {\it rooted binary 
tree representation} of $\mathcal{R}$ defined as follows. A 
vertex of this tree corresponds to some  
rectangle of $M_T$. Especially, the root vertex corresponds to $M_T$, and 
a leaf corresponds to a  rectangle in $\mathcal{R}$. 
For each non-leaf vertex $v$,  rectangles corresponding to 
its children consist of a partition of a  rectangle 
corresponding to $v$. 
Then, the size of a smallest set of disjoint monochromatic 
 rectangles which recursively partitions $M_T$ is
defied by $C^P(T)$, called the {\it protocol partition number} of $M_T$.

Given a Boolean function $f\colon \{0,1\}^n \to \{0,1\}$,
let $f^{-1}(1)$ (resp., $f^{-1}(0)$) be the set of 
$x \in \{0,1\}^n$ such that $f(x)=1$ (resp., $f(x)=0$).
For each Boolean function $f\colon \{0,1\}^n \to \{0,1\}$, 
we define the relation $T_f$ by 
\begin{equation*}
T_f = \{(x,y,i) \in  f^{-1}(1) \times f^{-1}(0) \times \{1,\ldots,n\} \mid x_i \neq y_i\}. 
\end{equation*}
(In order to avoid triviality, we assume $f^{-1}(1)\neq \emptyset$
and $f^{-1}(0)\neq \emptyset$.)
We are now ready to show the characterization of 
the size of a smallest formula presented by Karchmer and Wigderson~\cite{KW90}. 

\begin{theorem}[Karchmer and Wigderson~\cite{KW90}]
For each Boolean function $f$, 
\begin{equation*}
C^P(T_f)=L(f).
\end{equation*}
\end{theorem}

\subsection{The quasi-additive bound}
\label{section:quasi}

Here we introduce the {\it quasi-additive bound} presented by 
Ueno~\cite{U10}. Suppose that we are given a relation $T$ on 
$(X,Y,Z)$. We denote by $C_T$ the set of cells of $M_T$, i.e., $C_T=X \times Y$. 
Let $\mathcal{R}(T)$ be the set of  
rectangles of $M_T$, and let $\mathcal{M}(T)$ be the 
set of monochromatic  rectangles of $M_T$.   
For each $R \in \mathcal{R}(T)$, 
we denote by $\mathcal{P}(R)$ the set of partitions of $R$. 
Then, we consider the following linear programming for 
$\phi \in \mathbb{R}^{C_T}$ and 
$\psi \in \mathbb{R}^{C_T\times \mathcal{R}(T)}$. 
The objective is to 
maximize 
\begin{equation*}
\sum_{c \in C_T}\phi_c
\end{equation*}
under the constraints that 
\begin{equation*}
\sum_{c \in R}\phi_c + \sum_{c \in C_T \setminus R}\psi_{c, R} \le 1
\end{equation*}
for all $R \in \mathcal{M}(T)$, and 
\begin{equation*}
\sum_{c \in C_T\setminus V}\psi_{c, V} + \sum_{c \in C_T \setminus W}\psi_{c, W} 
\ge 
\sum_{c \in C_T\setminus R}\psi_{c, R}
\end{equation*}
for all $R \in \mathcal{R}(T)$ and $\{V,W\}\in \mathcal{P}(R)$. 
We denote by ${\bf LP}(T)$ this linear programming.
Let ${\bf QA}(T)$ be the optimal objective value of ${\bf LP}(T)$, 
and it is called the quasi-additive bound. 
Although ${\bf LP}(T)$ can be seen as a simple extension of the linear programming
for computing the rectangle bound, 
Ueno~\cite{U10} showed the following surprising result. 
\begin{theorem}[Ueno~\cite{U10}] \label{theorem:U10}
For each relation $T$, 
\begin{equation*}
{\bf QA}(T) = C^P(T),
\end{equation*} 
which implies that ${\bf QA}(T_f) = L(f)$
for each Boolean function $f$. 
\end{theorem}

\section{Main Results}

In this section, we use the same notations for a relation $T$ 
in Section~\ref{section:quasi}. 
For a relation $T$, let $\Gamma(T)$ be the set of $(R, P)$ such that 
$R \in \mathcal{R}(T)$ and $P \in \mathcal{P}(R)$, and 
we define the integer programming ${\bf PN}(T)$
for $x \in \mathbb{N}^{\mathcal{M}(T)}$ and 
$y \in \mathbb{N}^{\Gamma(T)}$ as follows. The objective 
is to minimize 
\begin{equation*}
\sum_{R\in \mathcal{M}(T)}x_{R}
\end{equation*}
under the constraints that  
\begin{equation} \label{eq1:ip}
\sum_{R \in \mathcal{M}(T) \colon c \in R} x_{R}=1
\end{equation}
for all $c \in C_T$ and 
\begin{equation} \label{eq2:ip}
\sum_{V\in \mathcal{R}(T)}
\sum_{\substack{P \in \mathcal{P}(V)\\\colon R \in P}} y_{V,P}
=
\left\{
\begin{array}{ll}
\displaystyle{\sum_{P \in \mathcal{P}(R)}y_{R, P}} + x_{R}, & \mbox{ if } R \in \mathcal{M}(T),
\\
\displaystyle{\sum_{P \in \mathcal{P}(R)}y_{R, P}}, & \mbox{ otherwise},
\end{array}
\right.
\end{equation}
for all $R \in \mathcal{R}^{\ast}(T)$, where 
$\mathcal{R}^{\ast}(T) = \mathcal{R}(T)\setminus \{C_T\}$. 
Ueno~\cite{Uphd} shown that the dual problem of the linear programming 
relaxation of ${\bf PN}(T)$ is equivalent to ${\bf LP}(T)$. 
Thus, in order to prove the main result, it suffices to show 
the following theorem.

\begin{theorem} \label{main_theorem}
For each relation $T$, the integer programming ${\bf PN}(T)$ 
computes the protocol partition number of $M_T$. 
\end{theorem}

Theorem~\ref{main_theorem} clearly follows from the 
following Lemmas~\ref{lemma1} and \ref{lemma2}. 
We say that $x \in \mathbb{N}^{\mathcal{M}(T)}$ 
is {\it feasible} to ${\bf PN}(T)$ 
if there exists $y \in \mathbb{N}^{\Gamma(T)}$ such that 
$(x,y)$ satisfies \eqref{eq1:ip} and \eqref{eq2:ip}. 
Notice that every element of 
$x \in \mathbb{N}^{\mathcal{M}(T)}$
which is feasible to ${\bf PN}(T)$ is 
$0$ or $1$ by the constraint \eqref{eq1:ip}. 

\begin{lemma} \label{lemma1}
Suppose that we are given a relation $T$ and 
a set $\mathcal{M}'$ of disjoint monochromatic  
rectangles of $\mathcal{M}(T)$ 
which recursively partitions $M_T$. Define $x \in \mathbb{N}^{\mathcal{M}(T)}$
by 
\begin{equation*}
x_{R}=
\left\{
\begin{array}{ll}
1, & R \in \mathcal{M}',\\
0, & \mbox{otherwise},
\end{array}
\right.
\end{equation*}
for each $R \in \mathcal{M}(T)$. 
Then, $x$ is feasible to ${\bf PN}(T)$. 
\end{lemma}
\begin{proof}
Since $\mathcal{M}'$ is a set of disjoint monochromatic 
rectangles which 
partitions $M_T$, $x$ clearly satisfies \eqref{eq1:ip}. 
Thus, it suffices to show that there exists 
$y \in \mathbb{N}^{\Gamma(T)}$
such that $(x,y)$ satisfies \eqref{eq2:ip}. 

Let $\mathcal{T}$ be a rooted binary
tree representation of $\mathcal{M}'$.  
In the sequel, we do not distinguish between 
a vertex $v$ of $\mathcal{T}$ and the  
rectangle to which $v$ corresponds. 
Define $y\in \mathbb{N}^{\Gamma(T)}$ by 
\begin{equation*}
y_{R,P}=
\left\{
\begin{array}{ll}
1, & \mbox{%
\begin{minipage}{68mm}
if $R$ is a non-leaf vertex of $\mathcal{T}$ and 
the children of $R$ consist of a partition $P$,
\end{minipage}
}\vspace{2mm}\\
0, & \mbox{otherwise},
\end{array}
\right.
\end{equation*}
for each $(R,P) \in \Gamma(T)$. 
Then, we show that $(x,y)$ satisfies \eqref{eq2:ip}. 

Let $R$ be a rectangle of $\mathcal{R}^{\ast}(T)$
which is not contained in $\mathcal{T}$. 
In this case, it follows from the definition of $y$ that 
$y_{R,P}=0$ for all $P \in \mathcal{P}(R)$ and
$y_{V,P}=0$ for all $(V,P) \in \Gamma(T)$ 
such that $R \in P$. 
Furthermore, even if $R \in \mathcal{M}(T)$, 
$x_{R}=0$ follows from $R \notin \mathcal{M}'$. 
These imply that \eqref{eq2:ip} satisfies.

Let $R$ be a rectangle of $\mathcal{R}^{\ast}(T)$ 
which is contained in $\mathcal{T}$.
Since $R\neq C_T$, $R$ is not the root of $\mathcal{T}$.
Hence, there exist the parent $V'$ and the sibling $S$  
of $R$ in $\mathcal{T}$. 
Using the notation $P'=\{R,S\}$, it follows from 
the definition of $y$ that $y_{V',P'}=1$ and 
$y_{V,P}=0$ for all $(V,P) \in \Gamma(T)$
such that $R \in P$ and $(V,P) \neq (V',P')$. 
Thus, the left-hand side of \eqref{eq2:ip} is equal to $1$, and 
it suffices to show that the right-hand side of 
\eqref{eq2:ip} is equal to $1$. 

If $R$ is a leaf of $\mathcal{T}$ (i.e., $R \in \mathcal{M}'$), 
$y_{R,P}=0$ for all $P \in \mathcal{P}(R)$ and 
$x_{R}=1$. Thus, the right-hand side of 
\eqref{eq2:ip} is equal to $1$.
In the case where $R$ is a non-leaf vertex of $\mathcal{T}$,
$x_{R}=0$ follows from $R \not\in \mathcal{M}'$. 
Let $P''$ be a partition of $R$ which consist of the children of $R$ in $\mathcal{T}$. 
Then, it follows from the definition of $y$ that 
$y_{R, P''}=1$ and $y_{R,P}=0$ for all 
$P \in \mathcal{P}(R) \setminus \{P''\}$.
These facts imply that the right-hand side of 
\eqref{eq2:ip} is equal to $1$.
This completes the proof.  
\end{proof}

\begin{lemma} \label{lemma2}
Suppose that we are given a relation $T$ and 
$x \in \mathbb{N}^{\mathcal{M}(T)}$ which is 
feasible to ${\bf PN}(T)$. 
Define $\mathcal{M}_x$ by 
\begin{equation*}
\mathcal{M}_x =
\{R \in \mathcal{M}(T)\mid x_R=1\}. 
\end{equation*}
Then, $\mathcal{M}_x$
is a set of disjoint monochromatic 
 rectangles of $\mathcal{M}(T)$ 
which recursively partitions $M_T$. 
\end{lemma}
\begin{proof}
For any relation $T$ and 
$x \in \mathbb{N}^{\mathcal{M}(T)}$ which is 
feasible to ${\bf PN}(T)$, it follows from 
\eqref{eq1:ip} that 
$\mathcal{M}_x$ is a set of disjoint monochromatic 
rectangles of $\mathcal{M}(T)$ 
which partitions $M_T$. Thus, what remains is 
to show that it {\it recursively} partitions $M_T$. 

For a relation $T$ and $x \in \mathbb{N}^{\mathcal{M}(T)}$, 
we say that $(T,x)$ is {\it eligible} if $x$ is feasible to ${\bf PN}(T)$. 
By induction on $|x|$, we show that the lemma holds for all eligible $(T,x)$.  
For all eligible $(T,x)$ such that 
$|x|=1$, the lemma holds 
since $\{M_T\}$ is 
a set of monochromatic  rectangles 
recursively partitions $M_T$. 

Assuming that the lemma holds for all eligible $(T,x)$ such that 
$|x| = k\ge 1$, we consider 
an eligible $(T,x)$ such that $|x| = k+1$. 
Let $y$ be a vector in $\mathbb{N}^{\Gamma(T)}$ 
such that $(x,y)$ satisfies \eqref{eq1:ip} and 
\eqref{eq2:ip}. 
For proving the lemma by induction, we first 
show the following claim. 
\begin{claim} \label{claim1}
There exists 
$(R', P') \in \Gamma(T)$ such that 
\begin{enumerate}
\item
every  rectangle in $P'$ is monochromatic, 
\item
$x_{V'}=1$ for all $V' \in P'$, and
\item
$y_{R', P'}>0$. 
\end{enumerate}
\end{claim}
\begin{proof}
Since $|x|\ge 2$, 
there exists $R \in \mathcal{M}(T)$ such that $x_{R}=1$ and 
$R \neq C_T$.
Hence, by \eqref{eq2:ip} 
there exists $(R,P) \in \Gamma(T)$ such that $y_{R,P}>0$. 
Let $(R', P')$ be a pair of $\Gamma(T)$ 
such that $y_{R',P'}>0$ and $|R'|$ is minimum.
Then, we can show that $(R', P')$ satisfies 
the above conditions as follows. 
If $V' \in P'$ is not monochromatic or $x_{V'}=0$,
it follows from \eqref{eq2:ip} that $y_{V',P}>0$ for some 
$P \in \mathcal{P}(V')$, which 
contradicts $|R'|$ is minimum. This completes the proof.
\end{proof}

Let $P'=\{V',W'\}$ be a pair of $\Gamma(T)$ satisfying the 
conditions of Claim~\ref{claim1}. 
Since $V'$ is monochromatic, there exists some index $i$ which 
every cell of $V'$ contains. 
Here we consider a new relation $T'$ obtained from $T$ by 
adding an index $i$ to the entry of every cell of $W'$.
Then, we define $x' \in \mathbb{N}^{\mathcal{M}(T')}$ by 
\begin{equation*}
x'_R = 
\left\{
\begin{array}{ll}
1, & \mbox{if } R = R', \\
0, & \mbox{if } R \in \{V',W'\}, \\
x_R, & \mbox{if } R \in \mathcal{M}(T) \setminus \{R', V',W'\}, \\
0, & \mbox{otherwise}, 
\end{array}
\right.
\end{equation*}
for each $R \in \mathcal{M}(T')$. 
Furthermore, we define $y \in \mathbb{N}^{\Gamma(T')}$ by 
\begin{equation*}
y'_{R,P} = 
\left\{
\begin{array}{ll}
y_{R,P} - 1, & \mbox{if } (R,P) = (R',P'), \\
y_{R,P}, & \mbox{otherwise}, 
\end{array}
\right.
\end{equation*}
for each $(R,P) \in \Gamma(T')$. 
Notice that $y'_{R',P'} \ge 0$
follows from $y_{R', P'}>0$.
Since $R' \notin \mathcal{M}(T)$ or 
$x_{R'}=0$ by \eqref{eq1:ip} and $x_{V'}=1$, 
we have $|x'|=k$.  
Hence, in order to 
use the induction hypothesis, we need the 
following claim. 
\begin{claim}
$(x',y')$ satisfies 
\eqref{eq1:ip} and \eqref{eq2:ip} for $T'$.
\end{claim}
\begin{proof}
Since \eqref{eq1:ip} is satisfied
by the definition of $x'$ and the induction 
hypothesis, we consider the constraint \eqref{eq2:ip}. 
By the definition of $(x', y')$ and 
induction hypothesis, it suffices to consider the constraint for 
$R'$, $V'$ and $W'$. 

First we consider the constraint for $R'$. 
Since $x'_{R'}-x_{R'}=1$ (if $R'$ is not contained in 
$\mathcal{M}(T)$, set $x_{R'}=0$) and 
\begin{equation*}
\sum_{P \in \mathcal{P}(R')}y'_{R',P}
-
\sum_{P \in \mathcal{P}(R')}y_{R',P}=-1, 
\end{equation*}
the right-hand side of \eqref{eq2:ip} does not change. 
Hence, since the left-hand side does not change, \eqref{eq2:ip}
is satisfied. 
Next we consider the constraint for $V'$. 
The left-hand side of \eqref{eq2:ip} decreases by $1$ due to $(R',P')$.
Since $x'_{V'} - x_{V'}=-1$, 
the right-hand side of \eqref{eq2:ip} also decreases by $1$. 
Hence, \eqref{eq2:ip} is satisfies. 
The same argument 
is clearly valid for $V'$. This completes the proof.
\end{proof}

By the induction hypothesis, $\mathcal{M}_{x'}$ recursively 
partitions $M_{T'}$. 
It is not difficult to see that 
we can construct a rooted binary tree 
representation of $\mathcal{M}_x$ by adding 
two vertices $V'$ and $W'$ under $R'$ of 
the rooted binary tree representation of $\mathcal{M}_{x'}$. 
This completes the proof. 
\end{proof}

Together with Theorem~\ref{theorem:U10} 
and the fact that the dual problem of the linear programming 
relaxation of ${\bf PN}(T)$ is equivalent to ${\bf LP}(T)$, 
the following main results of this note hold 
by Theorem~\ref{main_theorem}. 

\begin{corollary} \label{main_corollary}
For each relation $T$, ${\bf LP}(T)$ is the dual problem 
of the linear relaxation of the integer programming ${\bf PN}(T)$ 
for computing the protocol partition number of $M_T$.
\end{corollary}

\begin{corollary}
For each relation $T$, there exists no gap between 
the integer programming ${\bf PN}(T)$ 
for computing the protocol partition 
number of $M_T$ and its linear programming relaxation.
\end{corollary}


\end{document}